\documentclass[graybox]{svmult}

\usepackage{type1cm}        
\usepackage{makeidx}         
\usepackage{graphicx,multirow}        
\usepackage{multicol}        
\usepackage[bottom]{footmisc}

\usepackage{newtxtext}       %
\usepackage{newtxmath}       
\usepackage{amsmath,epsfig,amsfonts,dsfont}
\usepackage{cite}

\usepackage{enumitem}

\newcommand{\bx}{\boldsymbol{x}}
\newcommand{\by}{\boldsymbol{y}}
\newcommand{\ba}{\boldsymbol{a}}
\newcommand{\bc}{\boldsymbol{c}}
\newcommand{\bp}{\boldsymbol{p}}
\newcommand{\R}{\mathds{R}}
\newcommand{\bR}{\boldsymbol{R}}
\newcommand{\bY}{\boldsymbol{Y}}
\newcommand{\bh}{\boldsymbol{h}}
\newcommand{\bu}{\boldsymbol{u}}
\newcommand{\bv}{\boldsymbol{v}}
\newcommand{\bU}{\boldsymbol{U}}
\newcommand{\bV}{\boldsymbol{V}}
\newcommand{\bZ}{\boldsymbol{Z}}
\newcommand{\bL}{\boldsymbol{L}}
\newcommand{\bC}{\boldsymbol{C}}

\allowdisplaybreaks

\makeindex             

\begin{document}

\title*{Coupled oscillator networks for von Neumann and non von Neumann computing}
\author{Michele Bonnin, Fabio Lorenzo Traversa and Fabrizio Bonani}
\institute{Michele Bonnin \at Department of Electronics and Telecommunication, Politecnico di Torino, Corso Duca degli Abruzzi 24, 10129 Turin, Italy, \email{michele.bonnin@polito.it}
\and Fabio Lorenzo Traversa \at MemComputing Inc, 9909 Huennekens Street, Suite 110, 92121 San Diego, CA, United States \email{ftraversa@memcpu.com} \and Fabrizio Bonani \at Department of Electronics and Telecommunication, Politecnico di Torino, Corso Duca degli Abruzzi 24, 10129 Turin, Italy, \email{fabrizio.bonani@polito.it}}
%
%
\maketitle

\abstract{
The frenetic growth of the need for computation performance and efficiency, along with the intrinsic limitations of the current main solutions, is pushing the scientific community towards unconventional, and sometimes even exotic, alternatives to the standard computing architectures. In this work we provide a panorama of the most relevant alternatives, both according and not the von Neumann architecture, highlighting which of the classical challenges, such as energy efficiency and/or computational complexity, they are trying to tackle. We focus on the alternatives based on networks of weakly coupled oscillators. This unconventional approach, already introduced by Goto and Von Neumann in the 50s, is recently regaining interest with potential applications to both von Neumann and non von Neumann type of computing. In this contribution, we present a general framework based on the phase equation we derive from the  description of nonlinear weakly coupled oscillators especially useful for computing applications. We then use this formalism to design and prove the working principle and stability assessment of Boolean gates such as NOT and MAJORITY, that can be potentially employed as building blocks for both von Neumann and non-von Neumann architectures.      
}

\section{Introduction}

Electronic architectures designed to perform specialized or general computational tasks date back to the end of the first half of the 20th century, when the University of Pennsylvania disclosed ENIAC \cite{ENIACPatant,Goldstine1946ENIAC}. However, it is only few months after the ENIAC presentation, during the session of lectures at UPenn titled ``The Theory and Techniques for Design of Digital Computers'', that von Neumann wrote the {\it First Draft of a Report on the EDVAC} \cite{VNA} starting what we can dub the von Neumann architecture age.

For almost a century, the von Neumann architecture has been the standard reference for the design of electronic computing machines, especially for the general purpose ones. Its most basic description encompasses an input module, a central processing unit, a memory bank and an output module. Sets of instructions (programs) can be written in the memory, and the processing unit accesses the program and processes input data performing a sequence of operations, including reading and writing repeatedly on the memory, before ending the program and returning the output \cite{computer_architecture_book}.

The von Neumann architecture can be also viewed as the ideal hardware design of a Turing machine \cite{computer_architecture_book,computational_complexity_book}, and as such it inherits the related versatility along with the limitations. In fact, there is no known design that can be used to create a better artificial general purpose computing machine and, at the same time, be completely and deterministically controllable and programmable. However, the same design naturally reveals its drawbacks, such as the famous {\it von Neumann bottleneck} \cite{computer_architecture_book,78_Backus}, dictating that the system throughput is limited by the data transfer between  CPU and memory as the majority of the computation energy is used for the data movement rather than for the actual computation  \cite{computer_architecture_book,Hennessy2019,Horowitz2014}. Limitations do not come only from physics, but also from the computational scaling capability of these machines. For example, the von Neumann architectures are sequential in essence (notice that CPUs exploiting vectorization or distributed architectures like GPUs are still sequential in nature because, even if they can process a set of elements in parallel, the set size does not scale with the problem size) and they work with a finite set of symbols to encode information. On the one hand, both these features allow von Neumann architectures to handle virtually any decidable problem size (problems are termed decidable if they can be computed by a Turing machine \cite{computational_complexity_book}). On the other hand, many of these problems become quickly intractable for von Neumann architectures as their size grows. The theory of computational complexity based on Turing machines \cite{computational_complexity_book}, allows us to categorize decidable problems into classes of different polynomial scaling as a function of the problem size. Therefore, there are problems that can become intractable with von Neumann architectures either because the polynomial degree is too large for the typical problem size, or, even worse, if the scaling rule is a ``non-deterministic'' polynomial, as for the NP problem classes \cite{computational_complexity_book}, thus making them sometimes intractable even for very small size. For instance, solving a square system of linear equations in $n$ variables is a problem of complexity $\text{O}(n^3)$, i.e. its solution with a von Neumann architecture needs an order $n^3$ operations. Therefore, if $n=10^6$ we would need to perform an order  $10^{18}$ operations: this would require a time $\text{O}(10^9)$ s using a CPU working at GHz clock, i.e. tens of years! 

This picture highlights some among the most important reasons behind the quest of new solutions for computing devices, architectures and paradigms. The exponential growth of computation demand is everywhere, from industrial to consumer electronics, from cloud to edge, the internet of things, autonomous vehicles and much more. All of them require a day by day performance increase, along with energy efficiency and reliability. The demand is already challenging the current computing systems, the Moore's law era is at the end and projections for the near future clearly show that either new solutions will emerge, or the demand increase will soon become unsustainable  \cite{computer_architecture_book,Hennessy2019,Horowitz2014}. 

Recent years have seen many, novel or rediscovered, alternative solutions proposed by academia or industry giving some hope for satisfying near future computation demands. However, no clear winner is there yet. On the contrary, it looks like many specialized hardware and architectures are emerging, and there will be multiple winners depending on the applications or problems to be tackled. This is not surprising, since it is quite obvious that a price is paid for a machine to be general purpose. 

GPU is a nice example of a specialized computing architecture already introduced in the 1970s. Initially it was designed to accelerate specific graphic tasks by means of a distributed circuit design characterized by a very large memory bandwidth and a specialization in performing the algebraic operations needed for graphics, such as matrix-vector multiplication \cite{Nickolls2010}. In the last decade, GPUs are increasingly used for other type of applications, such as simulation and machine learning: industry and academia are investing a great deal of effort to develop more sophisticated, efficient and powerful GPUs, so that today they represent the most advanced form of high performance computing \cite{Nickolls2010,Singh2014}.       

The hype on artificial intelligence, recently invigorated by the renewed development of deep learning and artificial neural networks \cite{LeCun2015,Goodfellow2016}, is driving industry and academia towards new computing architectures specialized solely for the training or the inference of neural networks \cite{Sze2017,Sze2017a}. For these applications two main approaches exist: digital, such as Google's TPU \cite{Jouppi2017,Jouppi2018}, and analogue, as for instance the memristor crossbar neural networks on chip \cite{Li2018,Wang2019,Lee2020}. 

An example of a revived alternative suggested in the 1980s by Benioff, Feynman and Manin \cite{Benioff1980,Feynman1982,Manin1980} is quantum computing. The idea is to leverage quantum mechanics to create a non-Turing machine that redefines the computational complexity classes for some problems \cite{QI_bible}. For example, in 1994 Shor showed that, using entangled qubits and quantum gates for their manipulation, it is possible to set up a quantum circuit capable of factorizing integer numbers in cubic time with respect to the length (in digits) of the integer number \cite{Shor_0}. Obviously, this was a great theoretical achievement since it proves that a quantum circuit could break RSA encryption \cite{computational_complexity_book,Shor_0}. However, although quantum circuits seem a very appealing alternative, the technological barrier to their development has proved impervious: only in the last decade, few universities and organizations have accepted the challenge and are trying to build prototypes. Nevertheless, it is still in its infancy, and it will probably require decades before commercial-grade quantum computers will be available \cite{Dyakonov2019}.              

It is also worth mentioning the recent effort in neuromorphic computing, where a central role is played by asynchronous digital or analogue circuits trying to mimic brain neural networks \cite{Yu2018,Burr2016}. Noticeable examples are the systems like SpiNNaker developed at the university of Manchester \cite{Furber2014}, or Loihi by Intel Labs \cite{Davies2018}, or memristive systems at IBM \cite{Boybat2018} that implement spiking neural networks \cite{Tavanaei2019}.          

A general path aimed at mitigating the von Neumann bottleneck goes under the name of near-memory computing \cite{Singh2019}, with the limit case of in-memory computing \cite{Ielmini2018}. The idea here is to bring as close as possible memory and processing unit in order to reduce the energy and latency due to the data exchange between them \cite{Singh2019}. For in-memory computing, at least some logic is directly performed in memory \cite{Ielmini2018,DCRAM,Pershin2015,Sebastian2020}.    

A step forward beyond in-memory computing is the emerging technology dubbed memcomputing. It was introduced as a non-Turing paradigm \cite{UMM}, and it represents the most idealized model of a computational memory \cite{UMM,Di_Ventra2018,Pei2019}. It can be realized in analog \cite{traversaNP} or digital form \cite{DMM2}, however its most effective realization is through non-von Neumann architectures called self-organizing gates \cite{DMM2,Traversa2018a,manukian2017inversion}. The working principle of these gates is as follows: they are designed to reach an equilibrium such that the terminal states are consistent with some prescribed relation \cite{DMM2,Traversa2018a}. For example, self-organizing \textit{logic} gates reach an equilibrium when the terminals satisfy a prescribed set of Boolean relations \cite{DMM2}. However they are terminal agnostic, in the sense that they accept superposition of input and output signals at each terminal \cite{DMM2,Traversa2018a}. This allows to assemble self-organizing circuits by interconnecting self-organizing gates \cite{Di_Ventra2018,DMM2,Traversa2018a} with controllable properties of the dynamics \cite{no-chaos,noperiod}. This novel computing architecture has proven to be particularly efficient in solving combinatorial optimization problems \cite{Traversa2018,Sheldon2019,Traversa2019,Sheldon2019a,AcceleratingDL}.        

Almost all the emerging approaches so far mentioned have in common the exploitation of a transistor used as fundamental building block to create gates or other basic computing structures that ultimately reach some static working point. Alternatives to this scheme are for example the \textit{p}-bits introduced by Datta \textit{et al.,} \cite{Camsari2017}, that can be viewed as a stochastic alternative realization of memcomputing through CMOS-assisted nanomagnets \cite{Camsari2017} or magnetic tunnel junctions \cite{Camsari2017a}. A similar concept has been developed at Los Alamos National lab using artificial spin ice systems \cite{Caravelli2020}.         

This quest for alternative, non-conventional computing solutions has also recently revived the use of oscillators as building blocks for both von Neumann and non-von Neumann architectures. They were initially introduced independently by Goto \cite{Goto1954,Goto1959} and von Neumann \cite{vonNeumannPatent,Wigington1959} in the 1950s. The idea in this case is to encode information in the relative phase among different oscillators  and, by means of proper coupling, use them to perform basic operations such as the NOT or MAJORITY functions  \cite{Goto1954,Goto1959,vonNeumannPatent,Wigington1959}. In the 1960s, machines called Parametrons implementing the Goto design were built in Japan \cite{Goto1954,Goto1959}. Parametrons saw some successes, however they were soon eclipsed by digital computers exploiting the first generation of microscopic electronic devices. Today, the interest is renewed not only because there are many other modern and more compact ways to integrate oscillators, ranging from ring oscillators \cite{Farzeen2010} to spin-torque \cite{Houssameddine2007} and laser-based \cite{Kobayashi1980} structures and beyond \cite{Feng2008,Elowitz2000,Hannay2018,Matsuoka2011,13_amoeba,winfree1967}, but also because coupled oscillators potentially represent a very low power computing system when employed in von Neumann architectures \cite{roychowdhury2015,Raychowdhury2019,Bonnin2018}.   

However, interest in coupled oscillators is significant not only for von Neumann architectures, as recently they have been listed among the most promising building blocks for non-von Neumann architectures such as Ising machines \cite{Csaba2018,Wang2017,Chou2019} and bio-inspired computing architectures \cite{Mallick2020,Csaba2020,Chen2002,Shukla2014}.

In this work we focus on the general mathematical framework that is used to describe the dynamics of networks of oscillators employed to perform computation in the von Neumann as well as non von Neumann sense. We introduce general concepts such as self-sustained oscillations, limit cycles and stability to describe a single autonomous system, and then we extend these to (weakly) coupled oscillatory systems in Section~\ref{Fundamental theory} where we also introduce examples of network configurations, e.g. master-slave, useful for certain types of computations. In Section~\ref{phase equation} we introduce the phase equation through the concept of the isochron and we show how to write a general phase-amplitude equation system that describes the dynamics of a network of weakly coupled oscillators. In Section~\ref{boolean logic} we discuss how to design a simple network of oscillators to perform Boolean logic by means of a register and a complete base formed of the NOT and the MAJORITY gates. Moreover, we also show how to obtain stable OR and AND gates from the MAJORITY operation. For each case, we discuss the working principle showing the asymptotic stability properties of the phase equations describing the network. We finally draw our conclusions in Section~\ref{conclusions}.

\section{Basic unit, network architecture and computational principle}\label{Fundamental theory}

Oscillators are electronic circuits able to produce a continuous, repeated, alternating waveform without any input. Oscillators periodically convert energy between two different forms, and the energy dissipated during each oscillation cycle can, and should, be significantly less than the total energy involved in the oscillation. For example, an $LC$ oscillator characterized by a quality factor $Q$ loses approximately a fraction equal to $1/Q$ of its energy in each oscillation cycle. 

Every electronic oscillator requires at least two energy-storage elements and at least one \textit{nonlinear element.} Figure~\ref{figure1}~(a) shows the basic computing unit that we shall consider.  Since both inductor and capacitor are linear and passive (i.e., $L > 0$ and $C>0$), the resistive one-port $N_R$ must be active, i.e., its characteristic should be contained at least partially in the second and/or fourth quadrant of the $(v,i)$ plane for oscillation to be possible \cite{chua1987}. Besides, the characteristic curve linking $v$ and $i$ is required to be nonlinear. An example is shown in Figure~\ref{figure1}~(b): As the current through the resistor increases, the forward voltage increases rapidly until it reaches a maximum. As the current is further increased, the voltage starts decreasing, until it reaches a minimum, and then it starts increasing again. The region between the peak and the minimum is characterized by a negative differential resistance, and thus the resulting circuit is dubbed \textit{negative resistance oscillator. }

\begin{figure}
	\centering
	\includegraphics[width=80mm]{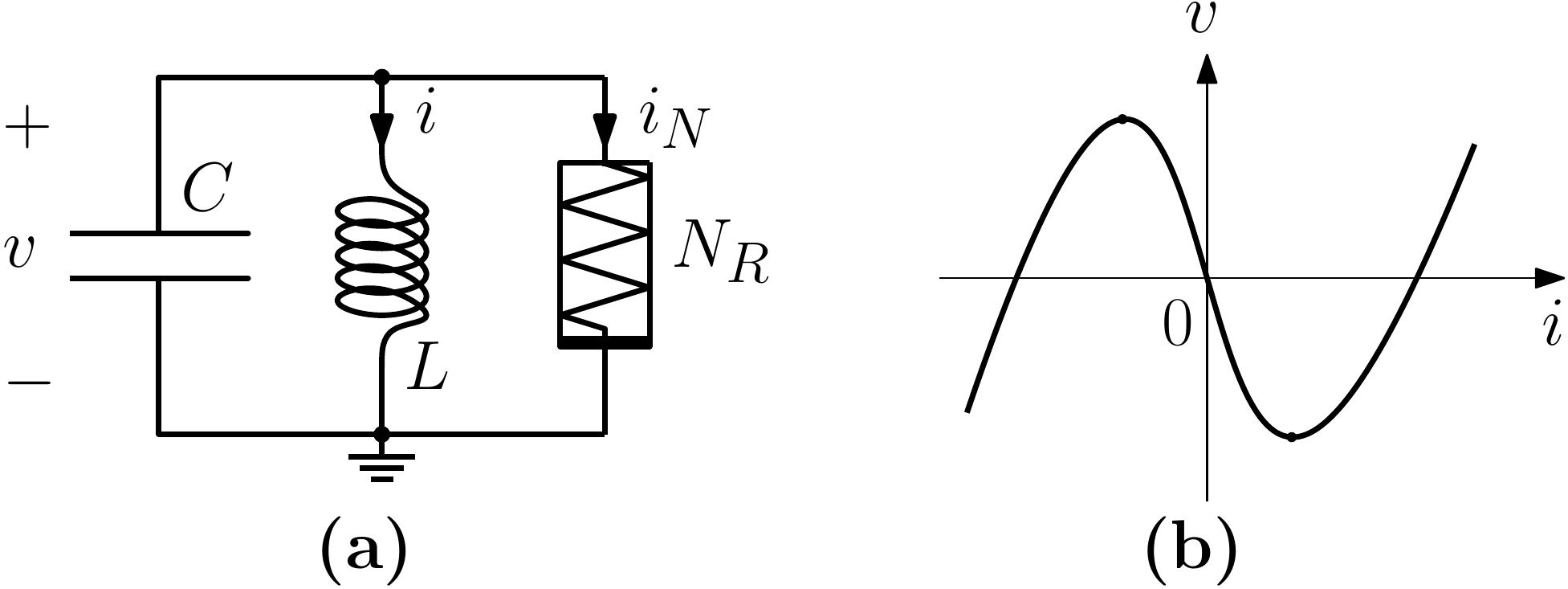}
	\caption{\textbf{(a)} Basic unit: $LC$ oscillator along with a nonlinear resistor. \textbf{(b)} Nonlinear $v(i)$ characteristic with a negative differential resistance region for $N_R$.}\label{figure1}
\end{figure}

The circuit exhibits {\em self-sustained oscillations} if the energy dissipated during the time in which it operates in the positive differential resistance regions, equals the energy supplied by $N_R$ when it operates in the negative differential resistance region. Self-sustained oscillations correspond to a \textit{limit cycle,} i.e. an isolated periodic orbit in the oscillator state space. Proving the existence of limit cycles is, in general, a very difficult if not impossible task. For the circuit in Figure~\ref{figure1}, it is however possible to find conditions for the existence of limit cycles. Applying Kirchhoff current law to one of the two nodes we obtain
\begin{equation}
	C \dfrac{\text{d} v}{\text{d} t}+ i_N + i = 0 \label{sec0-eq1}
\end{equation}
Deriving again with respect to time and using the nonlinear resistor characteristic $i_N = g(v)$,  
\begin{equation}
	\dfrac{\text{d}^2 v}{\text{d} t^2} + \dfrac{1}{C} \dfrac{\text{d} g(v)}{\text{d} v} \dfrac{\text{d} v}{\text{d}t} + \dfrac{1}{LC} v = 0 \label{sec0-eq2}
\end{equation}
Assuming a cubic nonlinearity $g(v) = -g_1 \, v + g_3 \, v^3$, where $g_1$ and $g_3$ are positive parameters, \eqref{sec0-eq2} is a {\em Li\'enard equation}, which is known to admit a stable limit cycle surrounding the origin \cite{perko2013}.

For many practical purposes, it is more convenient to rewrite the second order, ordinary differential equation \eqref{sec0-eq2}, as a system of two, first order differential equations
\begin{subequations}
	\begin{align}
		\dfrac{\text{d}i}{\text{d}t} & = \dfrac{1}{L} v \\[2ex]
		\dfrac{\text{d}v}{\text{d}t} & = - \dfrac{1}{C} i + \dfrac{1}{C} g_1 \, v - \dfrac{1}{C} g_3 v^3
	\end{align} \label{sec0-eq3}
\end{subequations}
Introducing the scaled variables 
\begin{equation}
	x = i, \qquad y = \sqrt{\dfrac{C}{L}} v \qquad \tau = \dfrac{1}{\sqrt{LC}} t \label{sec0-eq4}
\end{equation}
equation \eqref{sec0-eq3} becomes
\begin{subequations}
	\begin{align}
		\dfrac{\text{d}x}{\text{d}\tau} & = y \\[2ex]
		\dfrac{\text{d}y}{\text{d}\tau} & = -x + G(y) =  -x + G_1 y - G_3 y^3 
	\end{align} \label{sec0-eq5}
\end{subequations}
where $G_1 = g_1 \sqrt{L/C}$ and $G_3 = g_3 (L/C)^{3/2}$. 

For a wide range of values of parameters $G_1$ and $G_3$, the dynamical system \eqref{sec0-eq5} not only has a stable limit cycle, but it is also \textit{structurally stable,} meaning that the qualitative behavior of the trajectories is unaffected by small perturbations of the vector field. In particular the amplitude (but not the period) of the limit cycle is robust to external perturbations, showing small variations when perturbation terms are added to the vector field. This is a key requirement for computation purposes. 

To implement computation, oscillators must exchange information, i.e. they must be coupled to form a network. Connections between units must be weak, meaning that they should exchange small energy amounts, to avoid possible bifurcations of limit cycles. One basic assumption is that coupling is so weak that it affects only the phases of oscillators, without modifying significantly amplitude and frequency. A second major assumption to be made concerns how the oscillators are connected, that is, the network topology. The architecture plays a fundamental role in determining the network dynamics, and investigating the role played by interaction strength and topology is a rich and challenging research area (see \cite{albert2002,boccaletti2006} and references therein). On the one hand, the possibility to choose coupling schemes, interaction range and strength, gives the designer a wide range of possibilities to devise novel computational schemes. On the other hand, physical constraints must be taken into account. In a highly interconnected structure wirings may become burdensome, because the number of interconnects outweighs that of the neurons. Moreover, the mathematical  modelling and analysis complexity grows exponentially with the number of neurons and connections, thus making the problem quickly intractable.   

We shall consider networks with both local and global connections. In particular, the network has a central unit that will be used as a {\em reference}, connected to all other {\em neurons} as depicted in Figure~\ref{figure2}. A crucial point is that connections between the reference oscillator and neurons are one-directional, that is, the reference unit influences all other nodes, but not vice-versa. This kind of connection is also referred to as {\em master-slave}. The resulting graph is {\em directed} and {\em weighted}. Neurons have local connections with their first neighbors only. Such an architecture is inspired to biological neural structures, e.g. the thalamus-cortex system, and has been previously suggested as a promising topology for unconventional computing solutions \cite{hoppensteadt1999,itoh2004}.

\begin{figure}
	\centering
	\includegraphics[width=60mm]{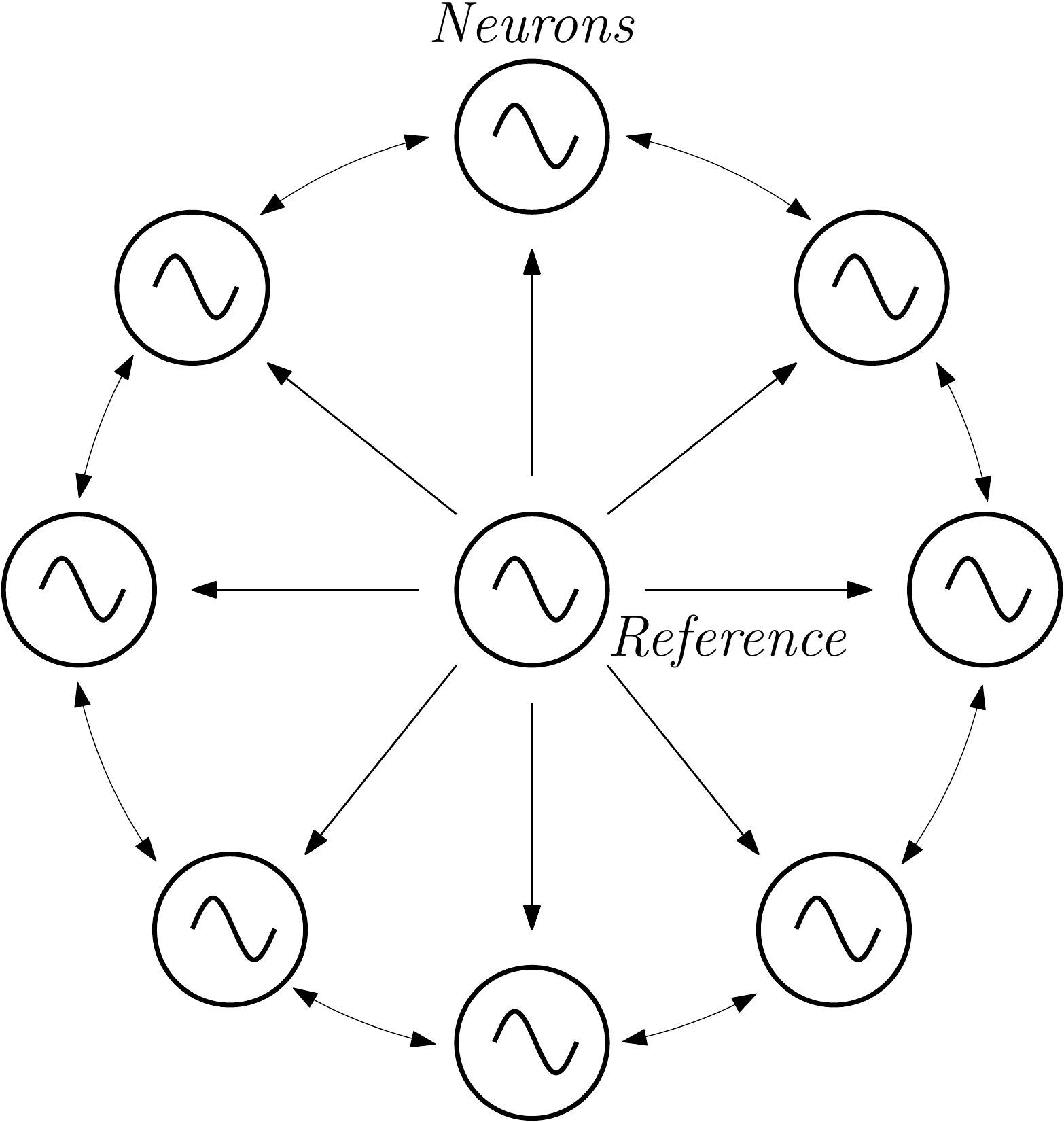}
	\caption{Network architecture.}\label{figure2}
\end{figure}

Information is encoded into the phase of each oscillator. Because the phase is a continuous variable, oscillatory networks in general operate as analog, or in modern terminology {\em bio-inspired}, computers. However, by a proper engineering of the couplings, it is possible to design oscillatory networks where only discrete values of the phase differences are attainable, i.e. those corresponding to asymptotically stable states, obtaining oscillatory networks that operate as digital computers based on Boolean logic.  

Information is processed modifying the phase relationships among  oscillators through the action of couplings. It should be noted that the simple presence of a physical connection between two oscillators does not guarantee by itself that they influence each other. Interaction occurs if the coupling strength exceeds a threshold that depends on the frequency mismatch. In other words, two oscillators whose free running frequencies are sufficiently apart, may be transparent one to the other, even if they are physically connected. Therefore it is possible, in principle, to control couplings without acting on the physical connections, but rather by simply modifying the free running frequencies of the individual oscillators. 

\section{Nonlinear oscillator networks and phase equation}\label{phase equation}

The most important concept to be defined for implementing computation with coupled oscillators is the phase concept. Defining the phase of a nonlinear oscillator is a nontrivial problem, whose analysis dates back to the seminal works \cite{guckenheimer1975,winfree1967}. Today, there is a large consensus that the most convenient way to define the phase of nonlinear oscillators is based on the concept of \textit{isochron} \cite{bonnin2012,bonnin2013,brown2004,ermentrout1996,kuramoto2003,nakao2016,teramae2009,wedgwood2013,wilson2016,yoshimura2008}.  

Consider a nonlinear oscillator with an asymptotically stable, $T$ periodic limit cycle $\bx_\text{s}$. Take a reference initial condition $\bx_\text{s}(0)$ on the limit cycle, and assign phase zero to this point, i.e. $\phi(\bx_\text{s}(0)) = 0$. The state of the oscillator with initial condition $\bx_\text{s}(0)$ at time $t$ is $\bx_\text{s}(t)$, with phase $\phi(\bx_\text{s}(t)) = 2\pi t/T = \omega \,t$. 

\begin{definition}
	The \emph{isochron} based at $\bx_\text{s}(0)$ is the manifold formed by the set of initial conditions $\bx_i(0)$ in the basin of attraction of the limit cycle, such that the trajectories leaving from $\bx_i(0)$ are attracted to $\bx_\text{s}(t)$, that is
	\begin{equation}
		I_{\bx_\text{s}(0)} = \left\{ \bx_i(0) \in \R^n/\bx_\text{s} \; \big| \; \lim_{t\rightarrow +\infty} ||\bx_i(t) - \bx_\text{s}(t) || = 0 \right\}.
	\end{equation}
\end{definition}
Isochrons represent the stable manifold foliation of $\bx_\text{s}(t)$ (see Figure~\ref{figure8}). The phase of the points belonging to the basin of attraction of the stable limit cycle is defined assigning the same phase to points lying on the same isochron. In this way, the phase of the point $\bx_i(t)$ results to be $\phi(\bx_i(t)) = \omega t + \phi(\bx_i(0))$. Thus the isochrons are the level sets of the scalar field $\phi(\bx)$.

\begin{figure}[tb]
	\begin{center}
		\includegraphics[width=40mm]{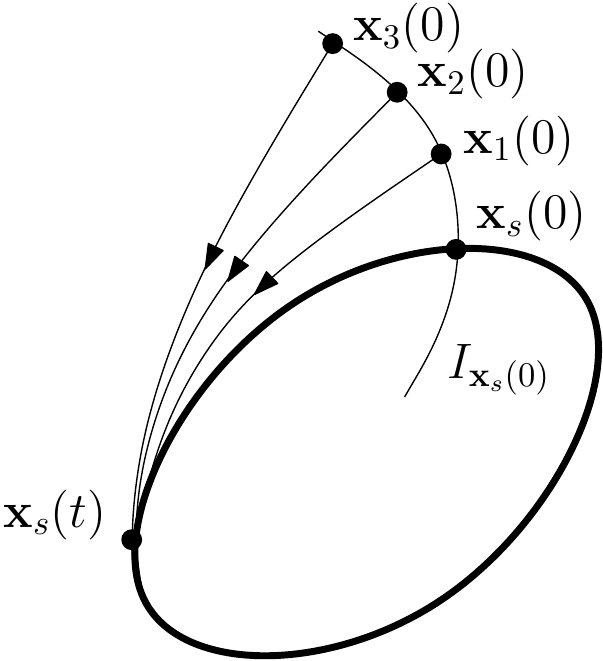}
		\caption{Definition of an isochron.}\label{figure8}
	\end{center}
\end{figure}

\begin{corollary}
	Isochrons are the invariant set for the Poincar\'e first return map $\Phi: I_{\bx_\text{s}(t)} \mapsto I_{\bx_\text{s}(t+T)}$. Different leaves are permuted by the flow, that is	
	\[ \bx_i(0) \in I_{\bx_\text{s}(0)} \Leftrightarrow  \bx_i(t) \in I_{\bx_\text{s}(t)} \]
\end{corollary}

\begin{proof}
	It follows directly from the definition of isochrons and the definition of the phase function.
\end{proof}

Isochrons and the phase function proved valuable tools to understand the physics of weakly perturbed oscillators, and the mechanisms leading to synchronization of coupled oscillators, but unfortunately their application to make quantitative predictions about real world oscillators is hindered by the fact that they can rarely be found analytically. Thus, one has to resort to rather sophisticated algorithms for their calculation \cite{guillamon2009,huguet2013,izhikevich2007,mauroy2012,mauroy2013,wilson2019,wilson2019b}.

An alternative solution is to decompose the oscillator's dynamics into two components, one tangent and one transversal (but not necessarily orthogonal) to the limit cycle \cite{wedgwood2013,aronson1996,bonnin2016,bonnin2017a,bonnin2017b,bonnin2019,demir2000,kaertner1990,traversa2015,wilson2020,wilson2020b}. The dynamics of the projection along the cycle represents the phase dynamics, whereas the transversal components are interpreted as the amplitude deviation. Obviously, both the phase and amplitude deviation variables depend upon the projection used, however it can be shown that, if the projection operators are chosen according to Floquet theory in the neighborhood of the limit cycle, the phase variable coincides with the phase defined using isochrons \cite{bonnin2017a,bonnin2017b}.

Consider a network composed by $N$ weakly coupled dynamical systems. The normalized state equations read
\begin{equation}
	\dfrac{\text{d} \bx_i}{\text{d} \tau} = \ba_i(\bx_i,\boldsymbol \mu_i) + \varepsilon \bc_i(\bx_1,\ldots,\bx_N) \qquad \textrm{for} \;\; i=1,\ldots,N \label{sec1-eq1}
\end{equation}
where $\bx_i: \R^+ \mapsto \R^n$ and $\ba_i: \R^n \mapsto \R^n$ are vector valued functions describing the state and the internal dynamics of the oscillators, respectively, $\boldsymbol \mu_i \in \R^m$ are parameter vectors, $\varepsilon \ll 1$ is a small parameter that defines the interaction strength, and $\bc_i: \R^{n\times N} \mapsto \R^n$ are vector valued functions describing the interactions. We shall assume that all functions are sufficiently smooth. 

Natural and even man made systems can never be perfectly identical. This can be accounted for assuming that the parameters $\boldsymbol \mu_i$ are subject to a  small deviation from an average value 
$\boldsymbol \mu_i = \boldsymbol \mu_0 + \delta \boldsymbol \mu_i$. Expanding the vector field around $\boldsymbol \mu_0$
\begin{equation}
	\ba_i(\bx_i,\boldsymbol \mu_i) = \ba_i(\bx_i, \boldsymbol \mu_0) + \dfrac{\partial \ba_i(\bx_i,\boldsymbol \mu_0)}{\partial \boldsymbol \mu_i} \delta \boldsymbol \mu_i + \ldots \label{sec1-eq2}
\end{equation}
where $\partial \ba_i(\bx_i,\boldsymbol \mu_0)/\partial \boldsymbol \mu$ is the matrix of partial derivatives of $\ba_i$ with respect to $\boldsymbol \mu$. For the sake of simplicity, we shall assume that the dynamic equations of the uncoupled systems are identical, and that the differences are small enough to approximate  \eqref{sec1-eq2} as
\begin{equation}
	\ba_i(\bx_i, \boldsymbol \mu_i) \approx \ba(\bx_i) + \varepsilon \bp_i(\bx_i) \label{sec1-eq3}
\end{equation}
Thus, the governing equations \eqref{sec1-eq1} can be rewritten as
\begin{equation}
	\dfrac{\text{d}\bx_i}{\text{d} \tau} = \ba(\bx_i) + \varepsilon \bp_i(\bx_i) + \varepsilon \bc_i(\bx_1,\ldots,\bx_N) \label{sec1-eq4}
\end{equation}
for $i=1,\ldots,N$.

It is convenient to introduce an ideal dynamical system, considered as reference
\begin{equation}
	\dfrac{\text{d} \bx}{\text{d} \tau} = \ba(\bx) \label{sec1-eq5}
\end{equation} 
It is important to stress that the reference system does not need to really exist, let alone belong to the considered set of oscillators. 
We restrict the attention to the case where the reference unit is an oscillator, i.e. we assume that \eqref{sec1-eq5} has an asymptotically stable $T$--periodic solution $\bx_\text{s}(\tau)$ such that
\begin{equation}
	\left \{ \begin{array}{c}
		\dfrac{\text{d} \bx_\text{s}(\tau)}{\text{d} \tau} = \ba(\bx_\text{s}(\tau)) \\[2ex]
		\bx_\text{s}(\tau) = \bx_\text{s}(\tau+T) \label{sec1-eq6}
	\end{array}\right.
\end{equation}
representing a limit cycle in the state space. The periodic solution is used to define the unit vector tangent to the limit cycle
\begin{equation}
	\bu_{1}(\tau) = \dfrac{\ba(\bx_\text{s}(\tau))}{|\ba(\bx_\text{s}(\tau))|} \label{sec1-eq7}
\end{equation}
where $|\cdot|$ denotes the $L_2$ norm. Together with $\bu_1(\tau)$, we consider other $n-1$ linear independent vectors $\bu_2(\tau),\ldots,\bu_n(\tau)$, such that the set $\{\bu_1(\tau),\ldots,\bu_n(\tau)\}$ is a basis for $\R^n$, for all $\tau$. Given the matrix $\bU(\tau) = [\bu_1(\tau),\ldots,\bu_n(\tau)]$, we define the reciprocal vectors $\bv_1^T(\tau),\ldots,\bv_n^T(\tau)$ to be the rows of the inverse matrix $\bV(\tau) = \bU^{-1}(\tau)$. Thus $\{\bv_1(\tau),\ldots,\bv_n(\tau)\}$ also spans $\R^n$ and the bi--orthogonality condition $\bv_i^T \bu_j = \bu_i^T \bv_j = \delta_{ij}$ for all $\tau$, holds. We shall also use  matrices $\bY(\tau) = [\bu_2(\tau),\ldots,\bu_n(\tau)]$, $\bZ(\tau) = [\bv_2(\tau),\ldots,\bv_n(\tau)]$, and the magnitude of the vector field evaluated on the limit cycle, $r(\tau) = |\ba(\bx_\text{s}(\tau))|$.

The use of Floquet bases for the sets $\{\bu_1(\tau),\ldots,\bu_n(\tau)\}$ and $\{\bv_1(\tau),\ldots,\bv_n(\tau)\}$ turns out to be very convenient, because they realize a complete decomposition of the state space. At each point $\bx \in \bx_s$, the limit cycle's stable manifold can be decomposed into two complementary linear spaces, the space $T_{\bx} M$ tangent to the limit cycle at $\bx$, and the space $T_{\bx} I$ tangent to the isochron at $\bx$. Extending the definition to all points belonging on the limit cycle leads to the tangent bundles $T M$ and $T I$, that are spanned by the vectors $\bu_1(\tau)$ and $\{\bu_2(\tau),\ldots,\bu_n(\tau)\}$, respectively. The covector $\bv_1(\tau)$ is normal to all vectors $\bu_k(\tau)$, for $k=2,\ldots,n$ and in turn, each covector $\bv_k(\tau)$, $k=2,\ldots,n$ is normal to $\bu_1(\tau)$. Thus $\bv_1(\tau)$ spans the one-dimensional cotangent bundle $N I$ orthogonal to $T I$, while the covectors $\{\bv_2(\tau),\ldots,\bv_n(\tau)\}$ span the $(n-1)$-dimensional cotangent bundle $N M$ orthogonal to $T M$ \cite{djurhuus2008,djurhuus2020}. The picture for a second order oscillator is illustrated in figure \ref{figure8bis}.

\begin{figure}[tb]
	\begin{center}
		\includegraphics[width=40mm]{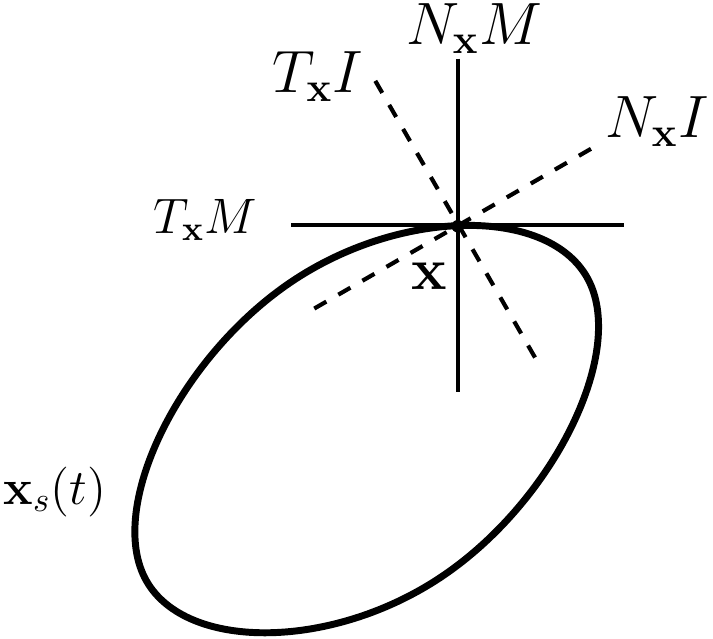}
		\caption{State space decomposition for a second order oscillator, illustrating the tangent spaces $T_{\bx} M$, $T_{\bx} I$, and the normal spaces $N_{\bx} M$, $N_{\bx} I$.}\label{figure8bis}
	\end{center}
\end{figure}

The following theorem establishes the equation governing the time evolution for the phase and amplitude deviation variables. 

\begin{theorem}[Phase-amplitude deviation equations\label{phase-amplitude model}]
	
	Consider system \eqref{sec1-eq4}, admitting of a $T$-periodic limit cycle $\bx_\text{s}(\tau)$ for $\varepsilon = 0$. Let $\{\bu_1(\tau),\ldots,\bu_n(\tau)\}$ and $\{\bv_1(\tau),\ldots,\bv_n(\tau)\}$ be two reciprocal bases such that $\bu_1(\tau)$ satisfies \eqref{sec1-eq7} and such that the bi--orthogonality condition $\bv_i^T \bu_j = \bu_i^T \bv_j = \delta_{ij}$ holds. Consider the coordinate transformation (dependence on $\tau$ in $\theta_i$ and $\bR_i$ is omitted for simplicity of notation)
	\begin{equation}
		\bx_i(\tau) = \bh(\theta_i,\bR_i) = \bx_\text{s}(\theta_i) + \bY(\theta_i) \bR_i \label{sec1-eq8}
	\end{equation}
	Then a neighborhood of the limit cycle $\bx_\text{s}(\tau)$ exists where the phase $\theta_i(\tau)$ satisfies
	\begin{equation}
		\dfrac{\text{d} \theta_i}{\text{d} \tau} = 1 +  a_{\theta}(\theta_i,\bR_i) + \varepsilon \, \delta a_{\theta_i}(\theta_i,\bR_i) + \varepsilon C_{\theta_i}(\theta_1,\bR_1,\ldots,\theta_N,\bR_N) \label{sec1-eq9} 
	\end{equation}
	with 
	\begin{subequations}
		\begin{align}
			&K(\theta_i,\bR_i)  = \left( |\ba(\bx_\text{s}(\theta_i))|  + \bv_1^T(\theta_i) \dfrac{\partial \bY(\theta_i)}{\partial \theta_i} \bR_i \right)^{-1}\label{sec1-eq10} \\[1ex]
			&a_{\theta}(\theta_i,\bR_i)  = K(\theta_i,\bR_i) \, \bv_1^T(\theta_i)  \, \bigg(\ba(\bx_\text{s}(\theta_i)+\bY(\theta_i) \bR_i) - \ba(\bx_\text{s}(\theta_i)) -\dfrac{\partial \bY(\theta_i)}{\partial \theta_i} \bR_i \bigg) \label{sec1-eq11} \\[1ex]
			&\delta a_{\theta_i}(\theta_i,\bR_i) =   K(\theta_i,\bR_i) \, \bv_1^T(\theta_i) \,  \bp_i\big( \bx_\text{s}(\theta_i) + \bY(\theta_i) \bR_i \big)\label{sec1-eq12} \\[1ex]
			&C_{\theta_i}(\theta_1,\bR_1,\ldots,\theta_N,\bR_N) \nonumber\\[1ex]
			&\qquad =  K(\theta_i,\bR_i) \, \bv_1^T(\theta_i) \,  \bc_i\big( \bx_\text{s}(\theta_1) + \bY(\theta_1) \bR_1,\ldots, \bx_\text{s}(\theta_N) + \bY(\theta_N) \bR_N \big) \label{sec1-eq13}
		\end{align}
	\end{subequations}	
	The amplitude deviations satisfy
	\begin{equation}
		\dfrac{\text{d} \bR_i}{\text{d} \tau} =  \bL(\theta_i) \bR_i +  \ba_{\bR}(\theta_i,\bR_i) + \varepsilon \, \delta \ba_{\bR_i}(\theta_i,\bR_i) + \varepsilon \, \bC_{\bR_i}(\theta_1,\bR_1,\ldots,\theta_N,\bR_N) \label{sec1-eq14} 
	\end{equation}
	where
	\begin{subequations}
		\begin{align}
			&\bL(\theta_i)  = -\bZ^T(\theta_i) \, \dfrac{\partial \bY(\theta_i)}{\partial \theta_i} \label{sec1-eq15} \\[1ex]
			&\ba_{\bR}(\theta_i,\bR_i)  = - \bZ^T(\theta_i) \, \left( \dfrac{\partial \bY(\theta_i)}{\partial \theta_i} \, \bR_i \, a_{\theta} (\theta_i,\bR_i) - \ba\big(\bx_\text{s}(\theta_i) + \bY(\theta_i) \bR_i \big) \right)  \label{sec1-eq16} \\[1ex]
			&\delta \ba_{\bR_i}(\theta_i,\bR_i) =   - \bZ^T(\theta_i) \, \left( \dfrac{\partial \bY(\theta_i)}{\partial \theta_i} \, \bR_i \, \delta a_{\theta} (\theta_i,\bR_i) - \bp_i\big(\bx_\text{s}(\theta_i) + \bY(\theta_i) \bR_i \big) \right) \label{sec1-eq17} \\[1ex]
			&\bC_{\bR_i}(\theta_1,\bR_1,\ldots,\theta_N,\bR_N) \nonumber\\[1ex]
			&\qquad=  \bZ^T(\theta_i) \,  \bc_i\big( \bx_\text{s}(\theta_1) + \bY(\theta_1) \bR_1,\ldots, \bx_\text{s}(\theta_N) + \bY(\theta_N) \bR_N \big) \label{sec1-eq18}
		\end{align}
	\end{subequations}	
\end{theorem}

\begin{proof}
	
	First, we show that the coordinate transformation \eqref{sec1-eq8} is invertible in the neighborhood of the limit cycle, and thus the problem is well posed. The Jacobian matrix of the coordinate transformation is
	\begin{equation}
		D\bh(\theta_i,\bR_i) = \left[\dfrac{\partial \bh}{\partial \theta_i}  \quad \dfrac{\partial \bh}{\partial \bR_i} \right] = \left[ \dfrac{\partial \bx_\text{s}(\theta_i)}{\partial \theta_i} + \dfrac{\partial \bY(\theta_i)}{\partial} \bR_i \quad \bY(\theta_i) \right] \label{sec1-eq19}
	\end{equation}
	On the limit cycle $\bR_i=0$ and then
	\[ D\bh(\theta_i,\bR_i)\big|_{\bR_i=0} = \left[\dfrac{\partial \bx_\text{s}(\theta_i)}{\partial \theta_i}  \quad \bY(\theta_i) \right] = \left [|\ba\big(\bx_\text{s}(\theta_i)\big)| \, \bu_1(\theta_i), \bu_2(\theta_i),\ldots, \bu_n(\theta_i)\right] \]
	Since $\{\bu_1(\tau),\ldots,\bu_n(\tau)\}$ is a basis for $\R^n$, it follows that the determinant of the Jacobian matrix is not zero. Then by the inverse function theorem there exists a neighborhood of $\bR_i=0$ where $\bh$ is invertible. Moreover, if $\bh$ is of class $\mathcal{C}^k$ then its inverse is also of class $\mathcal{C}^k$.
	
	Next we derive \eqref{sec1-eq9} and \eqref{sec1-eq14}. Taking the derivative of \eqref{sec1-eq8} we have
	\begin{align}
		\nonumber \dfrac{\text{d} \bx_i}{\text{d} \tau} & = \dfrac{\partial \bh(\theta_i,\bR_i)}{\partial \theta_i} \dfrac{\text{d}\theta_i}{\text{d} \tau} + \dfrac{\partial \bh(\theta_i,\bR_i)}{\partial \bR_i} \dfrac{\text{d}  \bR_i}{\text{d} \tau} \\[1ex]
		&\nonumber = \left( \dfrac{\partial \bx_\text{s}(\theta_i)}{\partial \theta_i} + \dfrac{\partial \bY(\theta_i)}{\partial \theta_i} \bR_i \right) \dfrac{\text{d}\theta_i}{\text{d} \tau} + \bY(\theta_i) \dfrac{\text{d} \bR_i}{\text{d} \tau} = \\[1ex] 
		&\nonumber = \ba\Big(\bx_\text{s}(\theta_i) + \bY(\theta_i) \bR_i \Big) + \varepsilon \bp_i\Big(\bx_\text{s}(\theta_i) + \bY(\theta_i) \bR_i \Big)  \\[1ex]
		&\qquad + \varepsilon \bc_i \Big( \bx_\text{s}(\theta_1)+ \bY(\theta_1) \bR_1, \ldots, \bx_\text{s}(\theta_N) + \bY(\theta_N) \bR_N \Big)  \label{sec1-eq20}
	\end{align}
	Multiplying to the left by $\bv_1(\theta_i)^T$, using the bi-orthogonality condition and rearranging terms we obtain \eqref{sec1-eq9}. Conversely, multiplying to the left by $\bZ^T(\theta_i)$ we find \eqref{sec1-eq14}. $\qed$
\end{proof}

Functions $\theta_i: \R^+ \mapsto \R^+$ can be interpreted as an elapsed time from the initial condition, representing a new parametrization of the trajectories. Because they are measured on the limit cycle, they map the interval $[0,T) \rightarrow [0,2\pi)$ and can be viewed as \emph{phase functions}. Functions $\bR_i : \R^+ \mapsto \R^{n-1}$ represent orbital deviations from the limit cycle, and they will be referred to as the \emph{amplitude functions} or \emph{amplitudes}, for simplicity. 

Figure \ref{figure3} shows the idea behind this decomposition. At any time instant the solution of the perturbed ($\varepsilon \ne 0$) system $\bx_i(t)$ can be decomposed into two contributions. The first component is the unperturbed limit cycle $\bx_\text{s}(\theta_i(t))$ evaluated at the unknown time instant $\theta_i(t)$, plus a distance $\bR_i(t)$ measured along the directions spanned by the vectors $\bu_2,\ldots,\bu_n$.

\begin{figure}
	\centering
	\includegraphics[width=50mm]{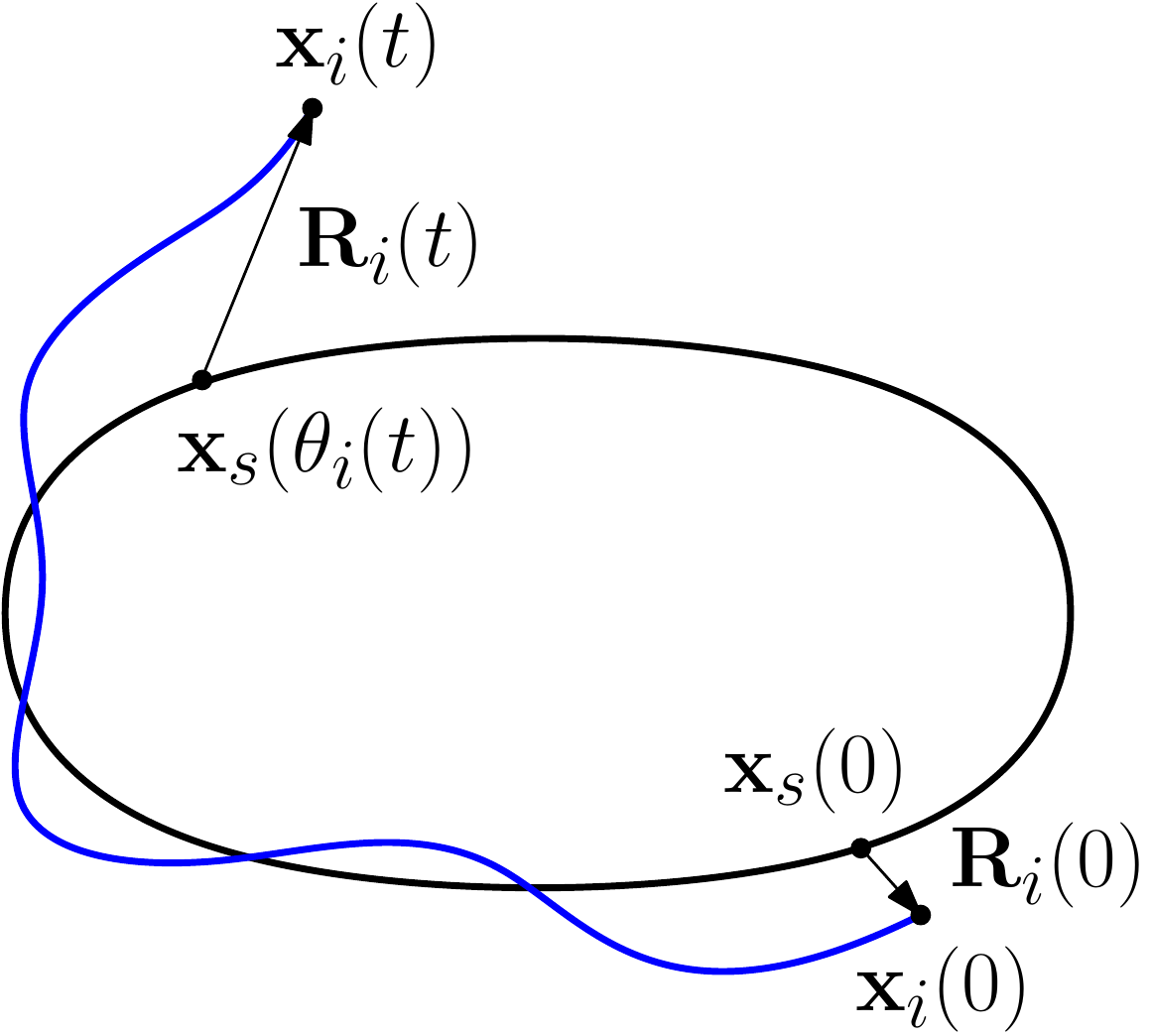}%
	\caption{State space decomposition of the solution $\bx_i(t)$}\label{figure3}
\end{figure}


For most practical applications, the phase is the most important variable. In fact, perturbations along the cycle are neither absorbed nor amplified, since this variable corresponds to the unitary Floquet multiplier and thus to the sole neutrally stable direction. They persist and become eventually unbounded asymptotically with time. Conversely,  perturbations transverse to the cycle are absorbed, being associated to the remaining Floquet multipliers (those with magnitude less than one for a stable limit cycle), as proved in the following corollary.

\begin{corollary}\label{corollary amplitude}
	Consider the reference oscillator \eqref{sec1-eq5} with the $T$-periodic limit cycle $\bx_\text{s}(\tau)$. Let $1,\mu_2,\ldots,\mu_n$ be the characteristic multipliers of the variational equation
	\begin{equation}
		\dfrac{\text{d}\by(\tau)}{d\tau} = \frac{\partial \ba(\bx_\text{s}(\tau))}{\partial \bx} \, \by(\tau) \label{sec1-eq22}
	\end{equation}
	Then $\bR_i=0$ is an equilibrium point for the amplitude equation of the unperturbed system,
	\begin{equation}
		\dfrac{\text{d} \bR_i}{dt} =  \bL(\theta_i) \bR_i + \ba_{\bR}(\theta_i,\bR_i)  \label{sec1-eq23}
	\end{equation}
	with characteristic multipliers $\mu_2,\ldots,\mu_n$.
\end{corollary}

\begin{proof}
	See \cite{bonnin2017a} theorem 2. 	
\end{proof}

The far reaching consequence of Corollary~\ref{corollary amplitude} is that, for weakly coupled oscillators that possess strongly attractive limit cycles, the phase dynamics can be effectively decoupled from the amplitude dynamics. In fact, for strongly stable limit cycles, it is reasonable to assume that the amplitude deviations remain  small, so that we can assume $\bR_i \approx 0$ and simplify \eqref{sec1-eq9} into
\begin{equation}
	\dfrac{\text{d} \theta_i}{\text{d} \tau} = 1 + \varepsilon \delta a_{\theta_i}(\theta_i) + \varepsilon C_{\theta_i}(\theta_1\ldots,\theta_N) \label{sec1-eq25}
\end{equation}
where 
\begin{align}
	\delta a_{\theta_i}(\theta_i) & = \dfrac{1}{|\ba(\bx_\text{s}(\theta_i)|} \, \bv_i^T(\theta_i) \, \bp_i\Big(\bx_\text{s}(\theta_i) \Big) \label{sec1-eq26}\\[2ex]
	C_{\theta_i}(\theta_1,\ldots,\theta_N) & = \dfrac{1}{|\ba(\bx_\text{s}(\theta_i)|} \, \bv_i^T(\theta_i) \, \bc_i \Big(\bx_\text{s}(\theta_1), \ldots, \bx_\text{s}(\theta_N)  \Big) \label{sec1-eq27}
\end{align}
Introducing the phase deviation $\psi_i = \theta_i - \tau$, representing the variation of the $i$-th oscillator phase with respect to the phase of the reference, and using \eqref{sec1-eq25}, we obtain
\begin{equation}
	\dfrac{\text{d} \psi_i}{\text{d} \tau} = \varepsilon \, \delta a_{\theta_i}(\psi_i + \tau) + \varepsilon \, C_{\theta_i}(\psi_1+\tau,\ldots,\psi_N+\tau) \label{sec1-eq28}
\end{equation}
For small values of $\varepsilon$, the phase deviation is therefore a slow variable. We can thus average integrating with respect to $\tau$ from 0 to $T = 2\pi$ without introducing a significant error,
obtaining the \emph{phase deviation equation}
\begin{equation}
	\dfrac{\text{d} \psi_i}{\text{d} \tau}   = \varepsilon \, \overline{\delta a}_{\theta_i}(\psi_i) + \varepsilon \, \overline C_{\theta_i}(\psi_1,\ldots,\psi_N) \quad \textrm{for} \; \; i=1,\ldots,N \label{sec1-eq29}
\end{equation}  
where
\begin{align}
	\overline{\delta a}_{\theta_i}(\psi_i) & = \dfrac{1}{2\pi} \int_0^{2\pi} \delta a_{\theta_i}(\psi_i+\tau) d\tau \label{sec1-eq30}\\[1ex]
	\overline C_{\theta_i}(\psi_1,\ldots,\psi_N) & = \dfrac{1}{2\pi} \int_0^{2\pi} C_{\theta_i}(\psi_1+\tau,\ldots,\psi_N+\tau) \, d \tau \label{sec1-eq31}
\end{align} 

Equilibrium points of the averaged phase deviation equation \eqref{sec1-eq29} corresponds to phase locked oscillations in the network of coupled oscillators. 

Thus, \eqref{sec1-eq29} can be used to design networks that exploit phase locked oscillations to perform fundamental logical operations.

\subsection{Example}

Before we consider the use \eqref{sec1-eq29} to design logic gates, we give an example of phase equation reduction at the single oscillator level. This will also serve as the starting point for further analysis. 

Consider the nonlinear oscillator described by \eqref{sec0-eq5}, that we rewrite in the form
\begin{subequations}
	\begin{align}
		\dfrac{\text{d}x}{\text{d}\tau} & = y \\[1ex]
		\dfrac{\text{d}y}{\text{d}\tau} & = - x + \alpha (1-y^2) y
	\end{align}\label{sec1-eq32}
\end{subequations}
where we have assumed $G_1 = G_3 = \alpha$. 

For second order systems, Floquet basis and co-basis can be computed using the formulas given in \cite{bonnin2012}, while for higher order systems, numerical methods are needed \cite{traversa2011,traversa2012,traversa2013}. Since the goal of this work is to illustrate fundamental principles avoiding technicalities as much as possible, we shall consider the weakly nonlinear version of \eqref{sec1-eq32} (i.e., we assume $\alpha \ll 1$). This permits to introduce some approximations allowing for the analytical determination of Floquet vectors and co-vectors. Moreover, these vectors take a particularly simple form that greatly simplifies the following discussion.

Introducing polar coordinates $x = A \cos \theta $, $y = A \sin \theta$, and averaging on the interval $[0,2\pi[$, \eqref{sec1-eq32} becomes
\begin{subequations}
	\begin{align}
		\dfrac{\text{d} \theta}{\text{d}\tau} = & 1\\[1ex]
		\dfrac{\text{d} A}{\text{d} \tau} = & \dfrac{\alpha}{2}A\left( 1 - \dfrac{3}{4}A^2\right) 
	\end{align} \label{sec1-eq33}
\end{subequations}
Equation \eqref{sec1-eq33} admits the asymptotic solution
\begin{equation}
	\bx_\text{s}(\tau) = \left( \begin{array}{c}
		\theta(\tau) \\
		A(\tau)
	\end{array} \right) = \left( \begin{array}{c}
		\tau \\
		2/\sqrt{3}
	\end{array} \right)
\end{equation}
Defining $\theta(\tau) = \tau \mod(2\pi)$, the solution is periodic and defines a limit cycle. The Jacobian matrix evaluated in the solution reads
\begin{equation}
	J|_{\bx_\text{s}} = \left( \begin{array}{cc} 
		0 & 0 \\
		0 & -\alpha
	\end{array} \right)
\end{equation}
The eigenvalues of the Jacobian matrix, $\lambda_1 = 0$ and $\lambda_2 = -\alpha$, are the cycle Floquet exponents. As expected, one exponent is null, the so called {\em structural exponent}, while the second one is negative, implying that the limit cycle is asymptotically stable. The eigenvectors of the Jacobian matrix $\bu_1 = [1,0]^T$ and $\bu_2 = [0,1]^T$, are the associated Floquet vectors. They form an orthogonal unitary basis, thus the co-vectors are $\bv_1 = \bu_1$ and $\bv_2 = \bu_2$. The Floquet vectors $\bu_1$, $\bu_2$ and co-vectors $\bv_1$, $\bv_2$ can be used to derive the amplitude-phase model and the phase deviation equation for a coupled network.

\section{Oscillator networks for Boolean logic}\label{boolean logic}

Boolean logic is based on binary variable states, let them be the ON and OFF states in electronic systems, 0 and 1 in computer science or $\left|\uparrow\right \rangle$ and $\left|\downarrow\right\rangle$ in quantum mechanics-based computations. In a coupled oscillators network the binary states correspond to in-phase and anti-phase locked oscillations, meaning that an oscillator must have the same phase (\emph{in-phase locking}), or have a phase difference equal to $\pi$ (\emph{anti-phase locking}) with respect to the reference oscillator.

\begin{figure}
	\centering
	\includegraphics[width=40mm,angle=-90]{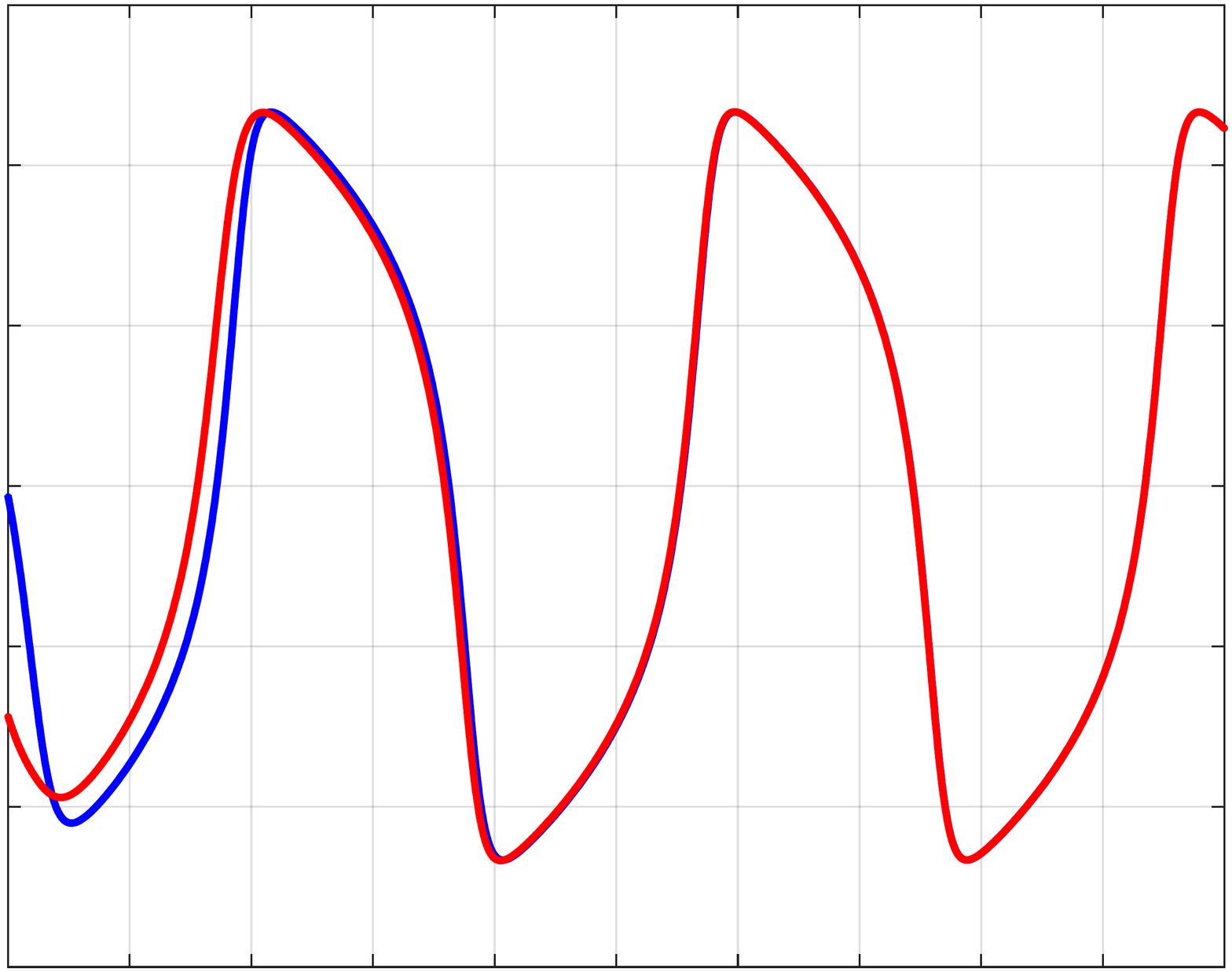}
	\includegraphics[width=40mm,angle=-90]{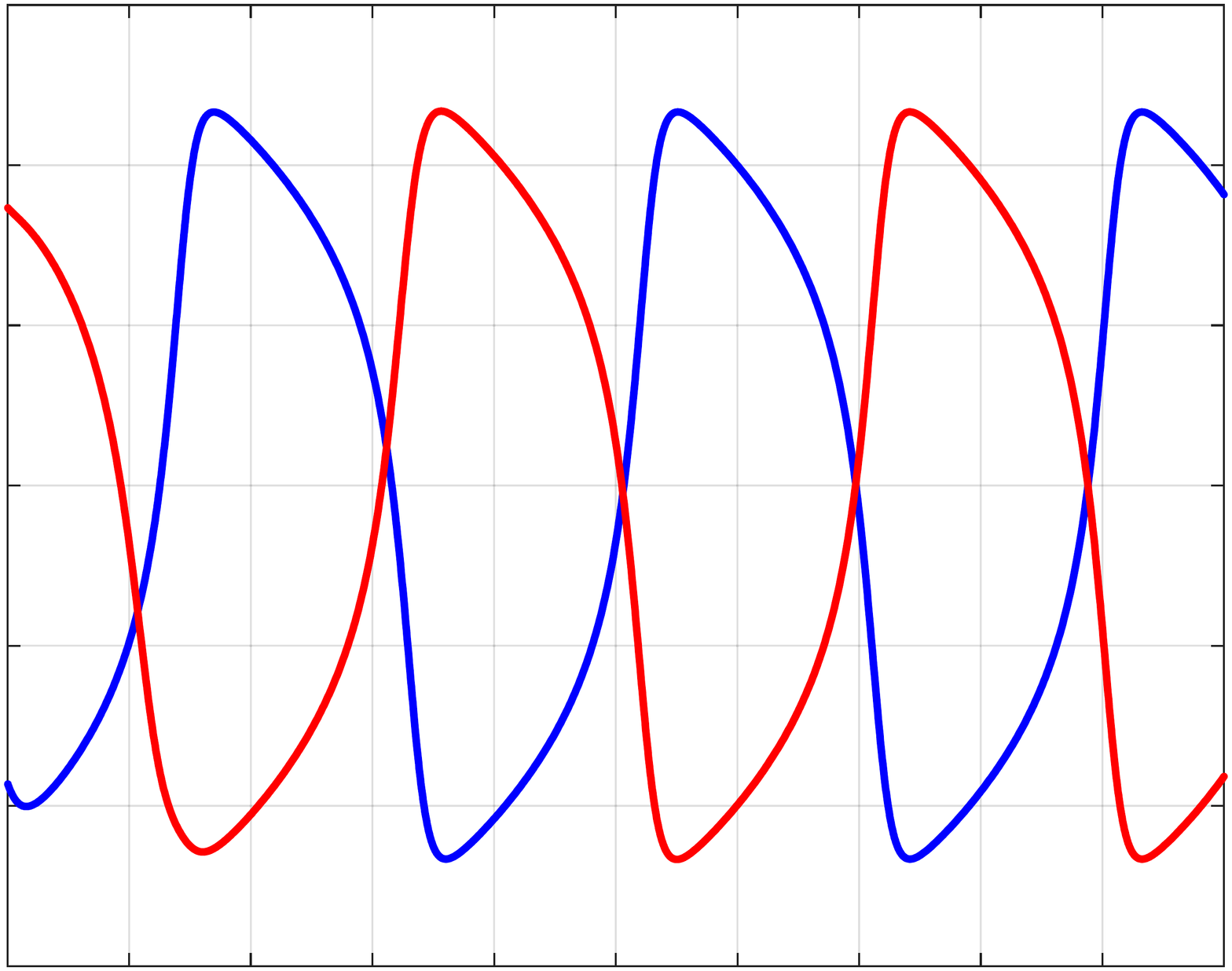}%
	\caption{Oscillator waveforms exhibiting phase locking after a short transient. On the left: In-phase locked waveforms. On the right: Anti-phase locked waveforms}\label{figure7}
\end{figure}

\subsection{Registers}
Information is stored into \textit{registers,} represented by oscillators connected to the central reference unit. As stated, the connection is one-directional, that is, the reference influences the registers, but not vice versa. By a proper engineering of the coupling, a register can be in-phase, or anti-phase, locked with the reference signal, corresponding to a stored datum equal to 0 or 1, respectively. 

Figure~\ref{figure4} shows a possible realization of a reference-register structure. The reference oscillator is represented by the controlled sources. Applying Kirchhoff laws the following state equations are obtained
\begin{subequations}
	\begin{align}
		\dfrac{\text{d} i_k}{dt} & = \dfrac{1}{L} v_k - \dfrac{R}{L} (i_k + i_R) \\[2ex]
		\dfrac{\text{d} v_k}{dt} & = -\dfrac{1}{C} - \dfrac{1}{C} g(v_k) - \dfrac{G}{C}(v_k - v_R)
	\end{align}\label{sec2-eq1}
\end{subequations} 
where $v_R$ and $i_R$ are the state variables of the reference, while $v_k$ and $i_k$ refer to the $k$-th register. We have also assumed that all oscillators are characterized by the same nonlinear conductance $g(v_k)$. After proper normalization
\begin{subequations}
	\begin{align}
		\dfrac{ \text{d}x_k}{\text{d} \tau} & = y_k - 2\rho (x_k + x_R) \\[2ex]
		\dfrac{\text{d} y_k}{\text{d} \tau} & = - x_k - g(y_k) - 2\gamma (y_k - y_R)
	\end{align}\label{sec2-eq2}
\end{subequations}
where $2\rho = R \sqrt{C/L}$ and $2\gamma = G \sqrt{L/C}$. Applying the procedure described in section \ref{phase equation} we obtain the phase deviation equation
\begin{equation}
	\dfrac{\text{d} \psi_k}{\text{d} \tau} = (\rho-\gamma) \sin \psi_k \label{sec2-eq3}
\end{equation}

\begin{theorem}[Register]\label{register}
	Consider the phase equation \eqref{sec2-eq3}. For $\gamma> \rho$, $\overline \psi_k = 0$ is an asymptotically stable equilibrium point, while $\overline \psi_k = \pi$ is unstable. For $\rho>\gamma$ the situation is reversed.   
\end{theorem}	

\begin{proof}
	Equation \eqref{sec2-eq3} has two equilibrium points, $\overline \psi_k = 0$ and $\overline \psi_k = \pi$.\\ 
	For $\rho > \gamma$, 	$\text{d} \psi_k/\text{d}\tau>0$ for $0<\psi<\pi$ and $\text{d} \psi_k/\text{d}\tau<0$ for $\pi<\psi<2\pi$, implying that $\overline \psi_k=0$ is unstable, while $\overline \psi_k=\pi$ is asymptotically stable. For $\rho<\gamma$ the situation is reversed. 
\end{proof}

Therefore, controlling the  value of parameters $\rho$ and $\gamma$ it is possible to switch the state of the register between the two digital values.

\begin{figure}
	\centering
	\includegraphics[width=80mm]{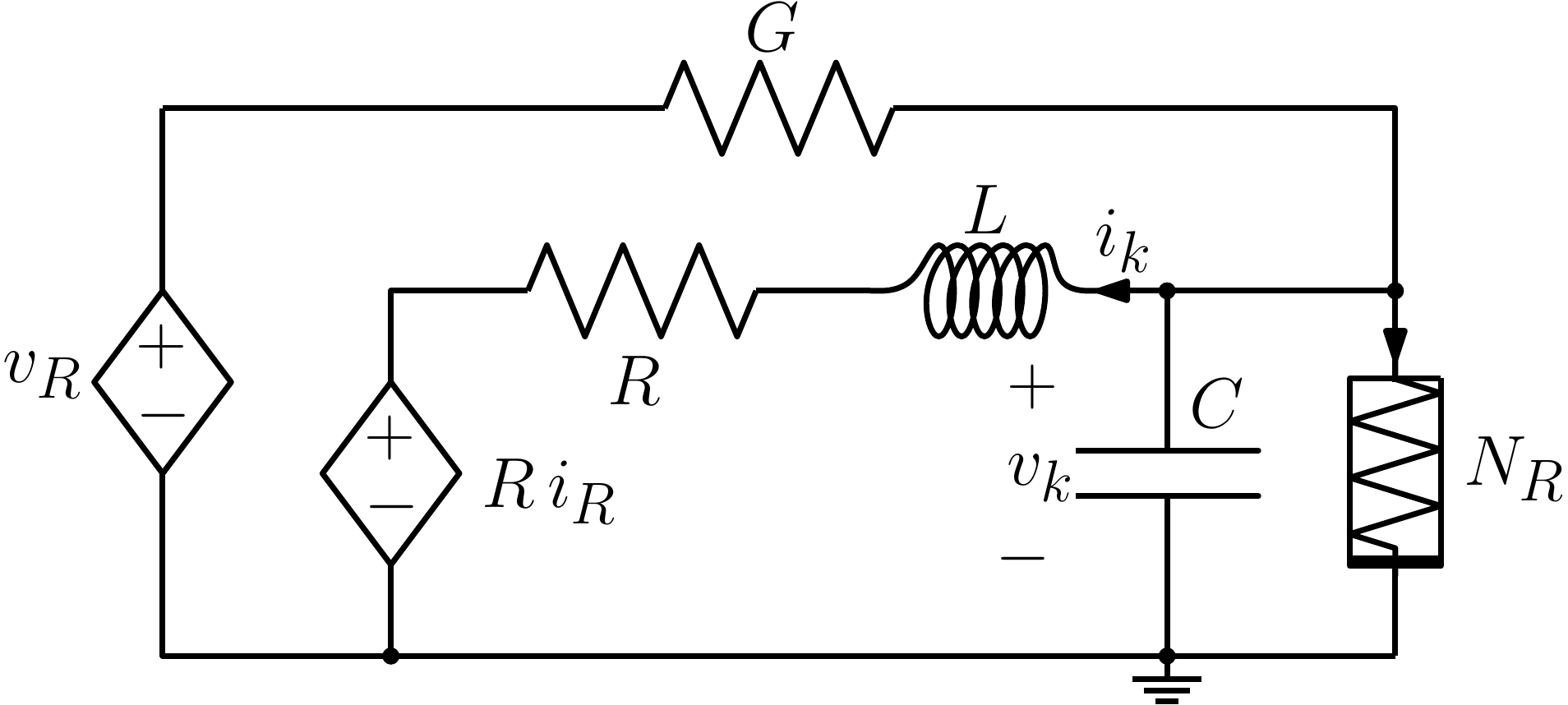}%
	\caption{Realization of the reference-register. Controlled sources represent the reference.}\label{figure4}
\end{figure}

\subsection{Logic gates}

Boolean computation requires the availability of a {\em logically} or {\em functionally complete} \cite{wernick1942,roychowdhury2015} set of elementary logical operations that can be composed to implement any combinational logic function. When logic states are encoded into phases, it is convenient to use the set composed by the NOT and the MAJORITY gates \cite{roychowdhury2015}.

\subsubsection{NOT gate}
A NOT gate is a single input, single output gate, that implements logical negation. An example of coupled oscillators realizing an invertible NOT gate is shown in figure~\ref{figure5}. 

\begin{figure}
	\centering
	\includegraphics[angle=-90,width=120mm]{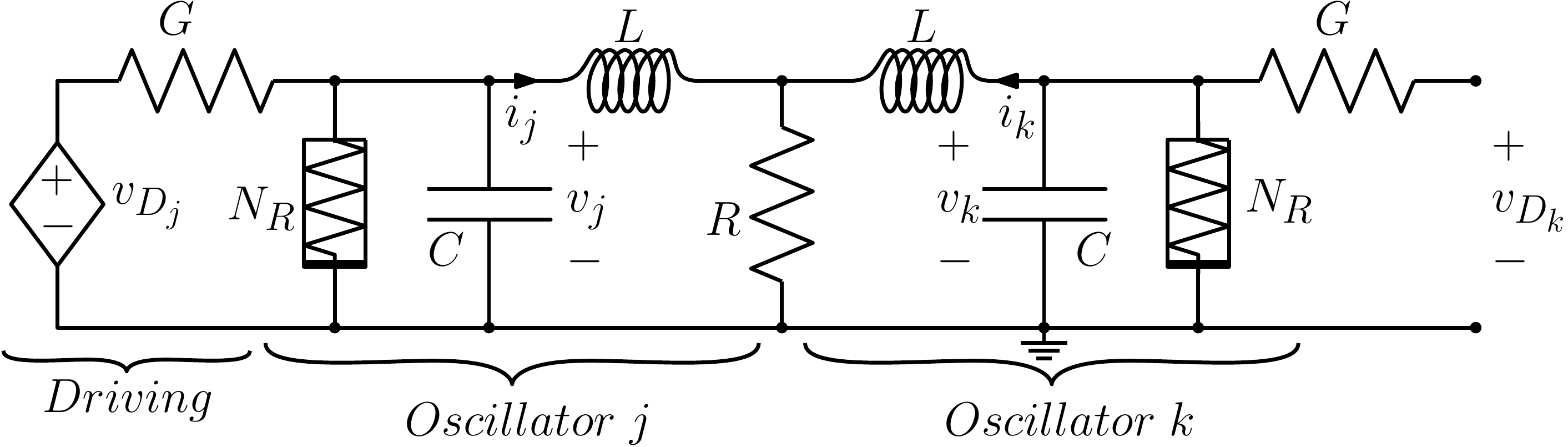}%
	\caption{Reversible NOT gate. The controlled voltage sources represent the driving signal used to select the input state.}\label{figure5}
\end{figure}

Two oscillators, denoted as $j$ and $k$, are coupled  through  resistor $R$. One oscillator represents the input, while the other is the output. The gate is denoted as {\em reversible} because the input and output oscillators can be interchanged without impairing the operation. For instance, in figure~\ref{figure5} oscillator $j$ is the input while oscillator $k$ is the output. The state of the input is controlled through the coupling with a driving oscillator unit, here represented by the controlled voltage source $v_{D_j}$, while the output is  voltage $v_{D_k} = v_k$.

Applying KVL and KCL, and  proper normalization, the state equations become
\begin{subequations}
	\begin{align}
		\dfrac{\text{d} x_j}{\text{d} \tau} & = y_j - 2\rho (x_j + x_k) \\[1ex]
		\dfrac{\text{d} y_j}{\text{d} \tau} & = -x_j - g(y_j) - 2\gamma (y_j - y_{D_j})\\[1ex]
		\dfrac{\text{d} x_k}{\text{d} \tau} & = y_k - 2\rho (x_k + x_j) \\[1ex]
		\dfrac{\text{d} y_k}{\text{d} \tau} & = -x_k - g(y_k) 
	\end{align}\label{sec2-eq4}
\end{subequations}
The corresponding phase deviation equations read
\begin{subequations}
	\begin{align}
		\dfrac{\text{d} \psi_j}{\text{d}\tau} & = \rho \sin(\psi_j-\psi_k) - \gamma \sin(\psi_j - \psi_{D_j}) \\[2ex]
		\dfrac{\text{d} \psi_k}{\text{d}\tau} & = \rho \sin(\psi_k-\psi_j) 
	\end{align}\label{sec2-eq5}
\end{subequations}

The following theorem proves that the phase deviation equations \eqref{sec2-eq5} corresponds to a NOT gate.

\begin{theorem}[NOT gate]\label{not gate}
	Consider the phase deviation equation \eqref{sec2-eq5}, with $\rho>0$ and $\gamma>0$. Let $\psi_{D_j} = 0$ (respectively $\psi_{D_j} = \pi$), then $(\overline{\psi}_j ,\overline \psi_k) = (0,\pi)$ (respectively $(\overline{\psi}_j, \overline \psi_k) = (\pi,0)$) is an asymptotically stable equilibrium point.
\end{theorem}

\begin{proof}
	It is trivial to observe that $(\overline \psi_{j},\overline \psi_{k}) = (0,\pi)$ and $(\overline \psi_{j},\overline \psi_{k}) = (\pi,0)$ are equilibrium points of \eqref{sec2-eq5}. The proof is based on the idea to show the existence of a strict Liapunov function in a neighborhood of the equilibrium point. The existence of such a Liapunov function implies that the equilibrium point is asymptotically stable.
	
	First, we consider the case $\psi_{D_j} = 0$, with equilibrium point $(\overline \psi_j, \overline \psi_k) = (0,\pi)$. Expanding one of the sine functions, we rewrite \eqref{sec2-eq5} in the form 
	\begin{subequations}
		\begin{align}
			\dfrac{\text{d} \psi_j}{\text{d}\tau} & = \rho \sin(\psi_j-\psi_k) -  \gamma \cos \psi_{D_j} \sin \psi_j  \\[2ex]
			\dfrac{\text{d} \psi_k}{\text{d}\tau} & = \rho \sin(\psi_k-\psi_j) 
		\end{align}\label{sec2-eq6}
	\end{subequations}
	
	Consider the function $V:[0,2\pi[ \times [0,2\pi[ \mapsto \R$
	\begin{equation}
		V(\psi_j,\psi_k) =  - \gamma \cos \psi_{D_j} \left( \cos \psi_j -1 \right) +  \rho \left( \cos(\psi_j - \psi_k) +1 \right)  \label{sec2-eq7}
	\end{equation}
	
	By definition, $V(\overline \psi_m,\overline \psi_n) = 0$, and $V(\psi_j,\psi_k) >0$ for all $\psi_j,\psi_k$ in a punctured neighborhood of $(\overline \psi_{j},\overline \psi_{k}) = (0,\pi)$.
	
	Finally, using definition \eqref{sec2-eq6}, we have
	\begin{equation}
		\dfrac{\text{d} V}{\text{d} \tau } = \dfrac{\partial V}{\partial \psi_j} \dfrac{\text{d} \psi_j}{\text{d} \tau} + \dfrac{\partial V}{\partial \psi_k}\dfrac{\text{d} \psi_k}{\text{d} \tau} = -  \left(\dfrac{\text{d} \psi_j}{\text{d} \tau} \right)^2 -  \left(\dfrac{\text{d} \psi_k}{\text{d} \tau} \right)^2 <0  \label{sec2-eq8}
	\end{equation}
	which proves that $V(\psi_j,\psi_k)$ is a strict Liapunov function in a neighborhood of $(0,\pi)$, as required.
	
	The case $\psi_{D_j}=\pi$, with equilibrium point $(\overline \psi_j,\overline \psi_k) = (\pi,0)$, is proved similarly, showing that the function
	\begin{equation}
		V(\psi_j,\psi_k) = - \gamma \cos \psi_{D_j} \left( \cos \psi_j + 1 \right) + \rho \left( \cos(\psi_j - \psi_k) +1 \right)  \label{sec2-eq9} 
	\end{equation}
	is a strict Liapunov function in the neighborhood of $(\pi,0)$.
	$\qed$
\end{proof}

\subsubsection{AND gate and OR gate}

A MAJORITY gate is a three inputs, one output logic gate that implements both the AND and the OR logic operations according to the truth table shown in table~\ref{tabella}. 

\begin{table}[tb]
	\centering
	\begin{tabular}{|c|c|c|c|c|}
		\hline
		& \multicolumn{1}{c|}{\bf{Reference}} & \multicolumn{1}{c|}{\bf{Input 1}} & \multicolumn{1}{c|}{\bf{Input 2}} & \multicolumn{1}{c|}{\bf{Output}} \\
		\hline
		\parbox[t]{2mm}{\multirow{3}{*}{\rotatebox[origin=c]{90}{\bf{AND}}}} & 0 & 0& 0& 0\\
		& 0 & 1 & 0 & 0\\
		& 0 & 0 & 1 & 0\\
		& 0 & 1 & 1 & 1\\
		\hline
		\parbox[t]{2mm}{\multirow{3}{*}{\rotatebox[origin=c]{90}{\bf{OR}}}} & 1 & 0& 0& 0\\
		& 1 & 1 & 0 & 1\\
		& 1 & 0 & 1 & 1\\
		& 1 & 1 & 1 & 1\\
		\hline
	\end{tabular}
	\vskip 1ex
\caption{Truth table for the MAJORITY gate.\label{tabella}}
\end{table}

An example of reversible MAJORITY gate implementation based on coupled oscillators is shown in figure \ref{figure6}. The first input is represented by the controlled voltage source $v_D$, which is applied to all oscillators and is used to set the function implemented by the gate. If the driving oscillator is set to one ($v_D$ is locked in-phase with the reference, in this case realized by the same driving circuit), then the MAJORITY gate realizes an OR gate. Conversely, if the driving register is set to zero ($v_D$ is anti-phase locked with the driving oscillator, this is realized applying a NOT to the reference), then the MAJORITY gate realizes an AND gate. This MAJORITY gate is called reversible because any pair of oscillators, i.e. $(i,j)$, $(i,k)$, or $(j,k)$, can be used as input, with the remaining oscillator representing the output. In this example, the pair $(i,j)$ defines the inputs, whilst  $k$ is the output. The input states are established applying driving signals $v_{D_i}$ and $v_{D_j}$ that can be set to one (in-phase locking with the reference) or zero (anti-phase locking). Voltage $v_{D_k} = v_k$ represents the output.
   
\begin{figure}[tb]
	\centering
	\includegraphics[angle=0,width=100mm]{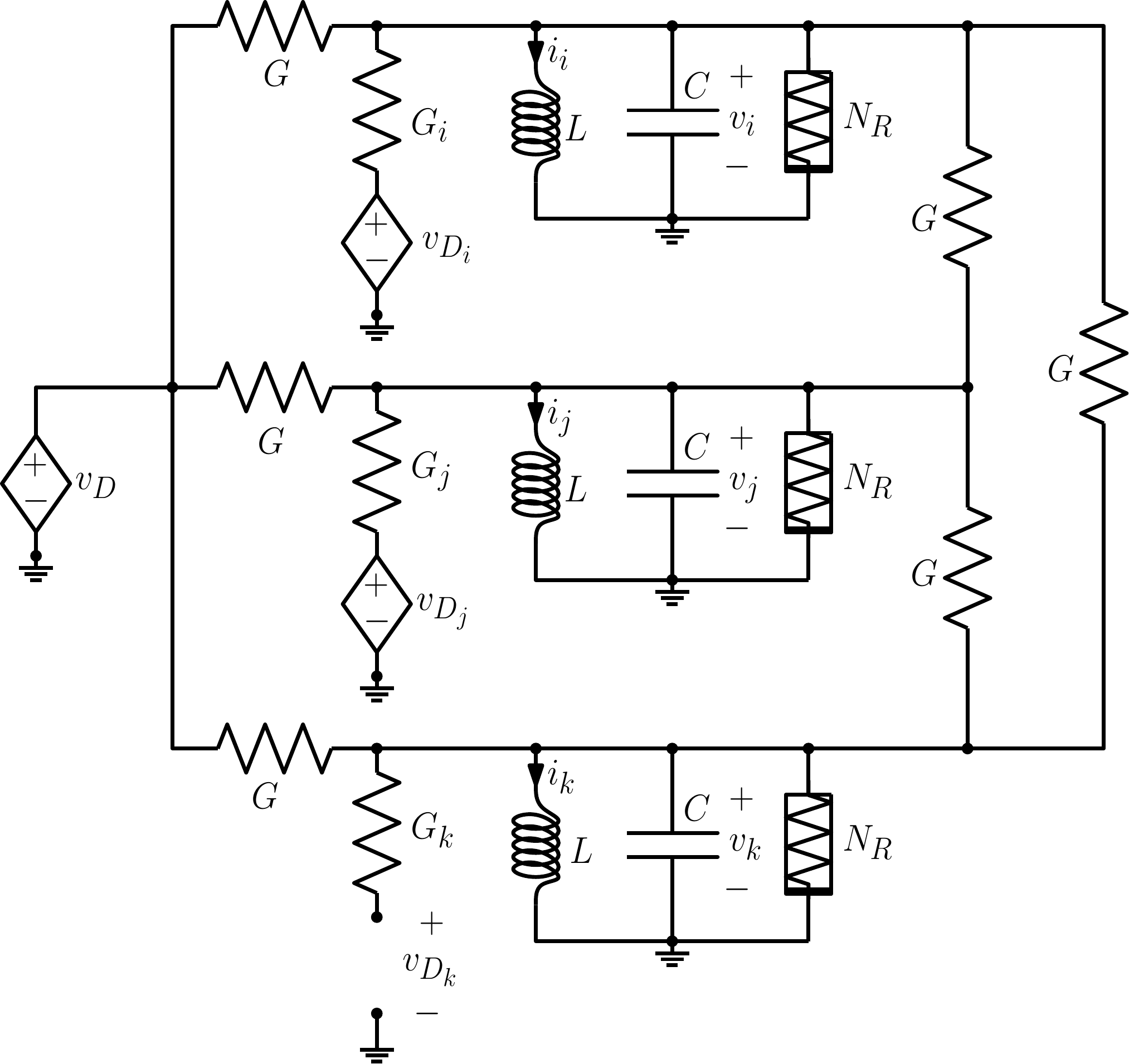}%
	\caption{Reversible MAJORITY gate. The controlled voltage sources represent the driving signals used to set the input state.}\label{figure6}
\end{figure}

Applying KVL and KCL, and with proper normalization, the state equations become
\begin{subequations}
	\begin{align} \label{sec2-eq10}
		\dfrac{\text{d}x_i}{ \text{d} \tau} & = y_i\\[2ex]
		\dfrac{\text{d} y_i}{\text{d} \tau} & = -x_i -g(x_i) - 2\gamma_i (y_i - y_{D_i}) - 2\gamma \sum_{m=j,k} (y_i-y_m ) - 2\gamma (y_i - y_D)\\[2ex]
		\dfrac{\text{d}x_j}{ \text{d} \tau} & = y_j\\[2ex]
		\dfrac{\text{d} y_j}{\text{d} \tau} & = -x_j -g(x_j) - 2\gamma_j (y_j - y_{D_j}) - 2\gamma \sum_{m=i,k} (y_j-y_m ) - 2\gamma (y_j - y_D)\\[2ex]
		\dfrac{\text{d}x_k}{ \text{d} \tau} & = y_k\\[2ex]
		\dfrac{\text{d} y_k}{\text{d} \tau} & = -x_k -g(x_k)  - 2\gamma \sum_{m=i,j} (y_k-y_m ) - 2\gamma (y_k - y_D)
	\end{align} 
\end{subequations}

After proper normalization, the resulting phase deviation equations read 
\begin{subequations}
	\begin{align}
		\dfrac{\text{d} \psi_i}{\text{d} \tau} & = - \gamma_i \sin (\psi_i - \psi_{D_i}) - \gamma \sin (\psi_i - \psi_D) - \gamma \sum_{m=j,k} \sin(\psi_i-\psi_m) \\[2ex]
		\dfrac{\text{d} \psi_j}{\text{d} \tau} & = - \gamma_j \sin (\psi_j - \psi_{D_j}) - \gamma \sin (\psi_j - \psi_D) - \gamma \sum_{m=i,k} \sin(\psi_j-\psi_m) \\[2ex]
		\dfrac{\text{d} \psi_k}{\text{d} \tau} & = - \gamma \sin (\psi_k - \psi_D) - \gamma \sum_{m=i,j} \sin(\psi_k-\psi_m)  
	\end{align} \label{sec2-eq11}
\end{subequations}
We are now in the position to prove that \eqref{sec2-eq11} realize both an AND and an OR gate.

\begin{theorem}[AND gate]\label{and gate}
	Consider the phase deviation equation \eqref{sec2-eq11}, with  $\psi_{D} = 0$ and $\gamma_i = \gamma_j > \gamma>0$. Then:
	\begin{enumerate}[label=(\alph*)]
		\item \quad If $\psi_{D_i} = \psi_{D_j} = 0$, then $(\overline \psi_i , \overline \psi_j , \overline \psi_k) = (0,0,0)$ is an asymptotically stable equilibrium point.
		
		\item \quad If $\psi_{D_i} = \pi$, $\psi_{D_j} = 0$ (respectively $\psi_{D_i} = 0$,  $\psi_{D_j} = \pi$), then $(\overline \psi_i, \overline \psi_j, \overline \psi_k) = (\pi,0,0)$ (respectively $(\overline \psi_i,\overline \psi_j ,\overline \psi_k) = (0,\pi,0)$) is an asymptotically stable equilibrium point.
		
		\item \quad If $\psi_{D_i} = \psi_{D_j} = \pi$ then $(\overline \psi_i, \overline \psi_j, \overline \psi_k) = (\pi,\pi,\pi)$ is an asymptotically stable equilibrium point.
	\end{enumerate}
\end{theorem} 

\begin{proof}
	Rewrite the phase deviation equations as
	\begin{subequations}
		\begin{align}
			\dfrac{\text{d} \psi_i}{\text{d}\tau} & = - (\gamma_i \cos\psi_{D_i} + \gamma \cos \psi_D) \sin \psi_i  - \gamma \left( \sin(\psi_i-\psi_j) + \sin(\psi_i - \psi_k) \right)  \\[2ex]
			\dfrac{\text{d} \psi_j}{\text{d}\tau} & = - (\gamma_j \cos\psi_{D_j} + \gamma \cos \psi_D) \sin \psi_j  - \gamma \left( \sin(\psi_j-\psi_i) + \sin(\psi_j - \psi_k) \right)  \\[2ex]
			\dfrac{\text{d} \psi_k}{\text{d}\tau} & = - \gamma \cos \psi_D \sin \psi_k - \gamma \left( \sin(\psi_k-\psi_j) + \sin(\psi_k - \psi_j) \right) 
		\end{align}
	\end{subequations}
	Consider the function $V:[0,2\pi[\times [0,2\pi[\times[0,2\pi[ \mapsto \R$
	\begin{align}
		\nonumber V(\psi_i,\psi_j,\psi_k)&  = - (\gamma_i \cos \psi_{D_i} + \gamma \cos \psi_D) (\cos \psi_i +N_i) \\[2ex]
		\nonumber & -  (\gamma_j \cos \psi_{D_j} + \gamma \cos \psi_D) (\cos \psi_j +N_j) -  \gamma(\cos \psi_k +N_k) \\[2ex]
		& -  \gamma (\cos(\psi_i - \psi_j) + \cos(\psi_i - \psi_k) + \cos(\psi_j - \psi_k) + N_{ij} + N_{ik} + N_{jk}) \label{sec2-eq13}
	\end{align}
	where
	\begin{equation}
		N_{\alpha} = \left\{\begin{array}{lcr}
			+1 & \textrm{if} & \overline \psi_{\alpha}=\pi \\[2ex]
			-1 & \textrm{if} & \overline \psi_{\alpha} = 0, 
		\end{array}\right. \qquad
		N_{\alpha \beta} = \left\{\begin{array}{lcr}
			+1 & \textrm{if} & \overline \psi_{\alpha} - \overline \psi_{\beta} = \pm \pi \\[2ex]
			-1 & \textrm{if} & \overline \psi_{\alpha} - \overline \psi_{\beta} = 0, 
		\end{array}\right. \label{sec2-eq14}
	\end{equation}
	By definition, $V(\overline \psi_i, \overline \psi_j , \overline \psi_k) = 0$, and it is easy to verify that $V(\psi_i,\psi_j,\psi_k)$ is positive definite for all $(\psi_i,\psi_j,\psi_k) \ne (\overline \psi_i,\overline \psi_j, \overline\psi_k)$, for each equilibrium point in (a), (b) and (c). Finally, computing the partial derivatives of \eqref{sec2-eq13} we have
	\begin{equation}
		\dfrac{\text{d} V}{\text{d}\tau} = \dfrac{\partial V}{\partial \psi_i} \dfrac{\text{d} \psi_i}{\text{d} \tau} + \dfrac{\partial V}{\partial \psi_j} \dfrac{\text{d} \psi_j}{\text{d} \tau} + \dfrac{\partial V}{\partial \psi_k} \dfrac{\text{d} \psi_k}{\text{d} \tau} = -  \left( \dfrac{\text{d} \psi_i}{\text{d} \tau} \right)^2 - \left( \dfrac{\text{d} \psi_j}{\text{d} \tau} \right)^2  - \left( \dfrac{\text{d} \psi_k}{\text{d} \tau} \right)^2 <0  
	\end{equation} 	
	for all $(\psi_i,\psi_j,\psi_k) \ne (\overline \psi_i,\overline \psi_j, \overline\psi_k)$. Thus \eqref{sec2-eq13} is a strict Liapunov function and the equilibria are asymptotically stable.	
	$\qed$
\end{proof}

\begin{theorem}[OR gate]\label{or gate}
	Consider the phase deviation equation \eqref{sec2-eq11}, with  $\psi_{D} = 0$ and $\gamma_i = \gamma_j > 2 \gamma>0$. Then:
	\begin{enumerate}[label=(\alph*)]
		\item \quad If $\psi_{D_i} = \psi_{D_j} = 0$, then $(\overline \psi_i , \overline \psi_j , \overline \psi_k) = (0,0,0)$ is an asymptotically stable equilibrium point.
		
		\item \quad If $\psi_{D_i} = 0$, $\psi_{D_j} = \pi$ (respectively $\psi_{D_i} = \pi$,  $\psi_{D_j} = 0$), then $(\overline \psi_i, \overline \psi_j, \overline \psi_k) = (0,\pi,\pi)$ (respectively $(\overline \psi_i,\overline \psi_j ,\overline \psi_k) = (\pi,0,\pi)$) is an asymptotically stable equilibrium point.
		
		\item \quad If $\psi_{D_i} = \psi_{D_j} = \pi$ then $(\overline \psi_i, \overline \psi_j, \overline \psi_k) = (\pi,\pi,\pi)$ is an asymptotically stable equilibrium point.
	\end{enumerate}
\end{theorem}

\begin{proof}
	The proof of (b) and (c) is similar to the proof of theorem~\ref{and gate}. To prove (a), we exploit local stability analysis. 
	
	For $\psi_D= \pi$, and $\psi_{D_i} = \psi_{D_j} = 0$, the phase deviation equation becomes
	\begin{subequations}
		\begin{align}
			\dfrac{\text{d} \psi_i}{\text{d}\tau} & = - (\gamma_i - \gamma) \sin \psi_i  - \gamma \left( \sin(\psi_i-\psi_j) + \sin(\psi_i - \psi_k) \right)  \\[2ex]
			\dfrac{\text{d} \psi_j}{\text{d}\tau} & = - (\gamma_i - \gamma ) \sin \psi_j  - \gamma \left( \sin(\psi_j-\psi_i) + \sin(\psi_j - \psi_k) \right)  \\[2ex]
			\dfrac{\text{d} \psi_k}{\text{d}\tau} & =  \gamma \sin \psi_k - \gamma \left( \sin(\psi_k-\psi_j) + \sin(\psi_k - \psi_j) \right) 
		\end{align}
	\end{subequations}
	with Jacobian matrix (we assume $\gamma_i = \gamma_j$)
	\begin{equation}
		J(\psi_i,\psi_j,\psi_k)\bigg|_{(0,0,0)} = \left( \begin{array}{ccc}
			-\gamma_i-\gamma & \gamma & \gamma \\[2ex]
			\gamma & -\gamma_i - \gamma & \gamma \\[2ex]
			\gamma & \gamma & - \gamma
		\end{array}\right)
	\end{equation}
	
	As the Jacobian matrix is symmetric, its eigenvalues are real. The characteristic polynomial takes the form
	\begin{equation}
		-P(\lambda) = \lambda^3 + (2\gamma_i + 3 \gamma) \lambda^2 + (\gamma_i^2 + 4 \gamma_i \gamma) \lambda + (\gamma_i^2 \gamma - 4 \gamma^3)
	\end{equation}
	Thus, by Descartes rule of signs, the eigenvalues are all negative. 
	$\qed$
\end{proof}

\section{Conclusions}\label{conclusions}

In this work we have discussed emerging alternatives to standard computing architectures aiming at overcoming their most challenging limits. This excursus  highlighted unconventional approaches in von Neumann architecture configurations as well as in non-von Neumann ones, and it culminated in the discussion of oscillator-based approaches. 

We have introduced a formal apparatus able to describe systems of weakly coupled oscillators, at the core of unconventional computing architectures based on autonomous systems. We have shown how to derive the phase equations that we can then use to assess the working principle and the stability of these systems. Moreover, as an example, we have shown how to design asymptotically stable NOT, MAJORITY, AND and OR gates by means of weakly coupled oscillators.

This work provides an important mathematical framework useful for the design and analysis of this emerging computing approach.

\bibliography{SUSYref.bib,Logic_gates}
\bibliographystyle{ieeetr}

\end{document}